\documentclass[preprint,amsmath]{revtex4}
\usepackage{graphicx,epsfig,dsfont,amssymb,amsthm,amsfonts,amsbsy,mathrsfs,amscd,appendix}  
\usepackage{booktabs}
\usepackage{enumerate}
\usepackage[colorlinks,linkcolor=blue,anchorcolor=blue,citecolor=blue]{hyperref}
\newtheorem{thm}{Theorem}

\newtheorem{lem}{Lemma}

\newtheorem{proposition}{Proposition}

\newcounter{appendixsection}
\renewcommand{\theappendixsection}{\Alph{appendixsection}}

\newcommand{\appendixsection}[1]{%
  \refstepcounter{appendixsection}%
  \section*{Appendix \theappendixsection. #1}%
  \addcontentsline{toc}{section}{Appendix \theappendixsection. #1}%
  \hypertarget{appsec:\theappendixsection}{}%
  \def\currentlabel{Appendix \theappendixsection.}%
}

\begin{document}

\begin{center}
		\bf{Trace-distance-based complementarity relations in a multipath interferometer}
       
\end{center}

\begin{center}
    Yue Sun$~^{1}$, Jingyan Liu$~^{2}$, Peng-Tong Li$~^{1}$, Chenxu Li$~^{3}$, Ming-Jing Zhao$~^{2,*}$
    
 \small $^{1}$School of Mathematics, Nanjing University of Aeronautics and Astronautics, Nanjing 210016, P. R. China \\
 \small $~^{2}$ {School of Science, Beijing Information Science and Technology University, Beijing 102206, P. R. China\\}
\small $^{3}$Center for Quantum Information, Institute for Interdisciplinary
Information Sciences, Tsinghua University, Beijing 100084, P. R. China\\
        
 \small $^{*}$Corresponding author: {zhaomingjingde@126.com}  
\end{center}

The complementarity relations in an interferometer reflect an important phenomenon in quantum mechanics: wave-particle duality. Here, we develop a  method to quantify
both wave and particle behaviors in a multi-path interferometer. In particular, we find that the trace  distance is a good candidate for wave and particle measures.  As a result, some duality relations and  triality ralations are established respectively.
This work not only extends the application of the trace distance to the interferometer, but also opens up new perspectives on the quantification of waveness and particleness.

\section{Introduction}
Wave-particle duality serves as one of the cornerstones of quantum mechanics, revealing the intrinsic nature of microscopic particles possessing both wave and particle features. It reveals that the wave and particle behaviors are mutually exclusive yet complementary.
For the two-path experiment, the first wave-particle duality in a quantitative way is \cite{zuizao,GY}
\begin{equation*}
P^{2}+V^{2}\leq 1,
\end{equation*} 
where $P$ measures particle behavior (which-path information), and $V$ measures wave behavior (interference). This formula shows a clear trade-off: knowing more about one aspect means knowing less about the other. 
In order to acquire further the path information, Englert introduced an interferometer model with path detectors, which led to  wave-particle duality relation \cite{Englert}
\begin{equation*}
D^{2}+V^{2}\leq 1,
\end{equation*}
where $D$ referred to as {\it path distinguishability}, quantifies the particle features in the two-path interferometer, and $V$ quantifies the visibility of the interference fringes.

Subsequently, quantitative characterization for the  wave and particle behaviors has attracted widespread attention, particularly in multi-path interferometers \cite{Ding,n1,n2,YS,TQ,GR,EK,Paul}. 
The first quantitative wave-particle duality relation in multi-path interferometers was proposed by D\"{u}rr, who also established criteria for the wave and particle measures \cite{durr}. With the development of quantum coherence, the wave measures based on coherence measures and the related wave-particle duality relations have been proposed, such as the $l_1$ norm \cite{MH}, relative entropy \cite{EB}, fidelity \cite{xiong}, skew information \cite{s3}, Fisher information \cite{n6}, etc.  In addition, numerous experimental studies have verified wave-particle duality from various aspects \cite{Biswas,Machado,Yoon}. 

It is true that all coherence measures meet the criteria for the wave measure and can be viewed as a wave measure formally. Hence some coherence-based methods are generally proposed to construct wave measures \cite{n3,du1}. However the coherence and waveness are not equivalent. For example, the $l_2$ norm (and more generally, the $l_p$-norm) which is not monotonic under incoherent operations is indeed a common-used wave measure \cite{durr,Tao}.
The relation between coherence and waveness is factually rather ambiguous. In the context, we intend 
to characterize the wave behavior from the perspective that is irrelevant to the coherence under the framework of incoherent operations.  On this basis, we develop a method to construct wave and particle measures,
including the relative entropy and the trace distance as special cases. 
As we know, the trace distance is generally not a coherence measure. However this work highlights its ability to quantify the wave behavior as well as the particle behavior in the multi-path interferometer.

To reveal the constraint relation between waveness and particleness, the wave-particle dualities in terms of the trace distance are established. In this process, we also find the role of mixedness in the multi-path interferometer.  Therefore, the wave-particle-mixedness triality are established further. Considering the extensive applications of the trace distance, such as quantum
cryptographic protocols \cite{Ben,Renner}, quantum complexity theory \cite{Watrous}, quantum state discrimination {\cite{MH}}, entanglement measure {\cite{Eisert}}, the application of the trace distance in the multi-path interferometer may be a potential link between different physical scenarios.


The structure of this paper is organized as follows. In Section \ref{secII}, we introduce the necessary background and develop a method to quantify both wave and particle behaviors in a multi-path interferometer. In Section \ref{sec III}, we study the complementarity relations based on trace distance in a multi-path interferometer, in which four wave-particle dualities, three wave-particle-mixedness trialities, and one generalized triality are established, respectively. Conclusions are made in Section \ref{sec IV}.

\section{Wave and Particle Measures in a multi-path interferometer}\label{secII}

In this section, we consider a \(d\)-path interferometer, as illustrated in Fig. \ref{fg1}. Let the orthonormal basis state $\{|i\rangle\}_{i=1}^d$ corresponds to the $d$ possible paths in the interferometer.
For any particle  enters  \(d\)-path interferometer, after passing through a beam splitter, it can be described by 
\begin{equation}
  \rho = \sum_{i,j=1}^{d} \rho_{ij} |i\rangle\langle j|, 
\end{equation}
with $\rho_{ij} = \langle i|\rho|j\rangle $ under the path basis $\{|i\rangle\}_{i=1}^d$.
In the interferometer, the particle exhibits both wave and particle behavior. The particle behavior is characterized by the path information, which is determined by the 
probability distribution $(\rho_{11}, \rho_{22}, \cdots, \rho_{dd})^T$. 
The wave behavior is characterized by the interference fringe visibility on the screen, which is determined by the superposition passing through different paths.

\begin{figure}[htbp]
    \centering
    \includegraphics[width=0.45\textwidth]{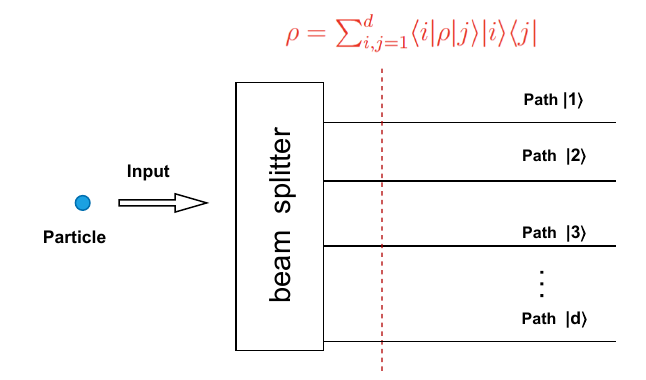}
    \caption{Schematic of a $d$-path interferometer.}
    \label{fg1}
\end{figure}

To quantify the particleness of the particle in the multi-path interferometer, the requirement for the measure of particleness \(\mathcal{P}(\rho)\) is proposed as  \cite{EK,durr,n3,du1}
\begin{enumerate}[(P1)]
\item \(\mathcal{P}(\rho)\) attains its global minimum when the particle's path is completely uncertain, i.e. $\langle i|\rho| i\rangle=1/d$ for all $i$.
    \item \(\mathcal{P}(\rho)\) attains its global maximum when the particle's path is completely certain, i.e. $\langle i|\rho| i\rangle=1$ for some $i$.
    \item \(\mathcal{P}(\rho)\) is invariant under the permutations of the path labels.
    \item \(\mathcal{P}(\rho)\) is convex.
\end{enumerate}

Likewise, to quantify the waveness of the particle in the multi-path interferometer, the requirement for wave measure \(\mathcal{W}(\rho)\) is proposed as \cite{EK,durr,n3,du1}
\begin{enumerate}[(W1)]
    \item $\mathcal{W}(\rho)$ attains its global minimum
   when $\rho$ is diagonal in the path basis. 
    \item $\mathcal{W}(\rho)$ attains its global maximum {when $\rho$ is pure state with equal diagonal entries i.e.  $\langle i|\rho|i\rangle = \frac{1}{d}$ for all $i$ }.
    \item $\mathcal{W}(\rho)$ is invariant under the permutations of the path labels.
    \item $\mathcal{W}(\rho)$ is convex.
\end{enumerate}

In fact, the mutual exclusivity between waveness and particleness is reflected by the requirements of the wave measure and particle measure. When the waveness in the interferometer attains the maximum, then the 
particle is in  the maximally coherent pure state 
$|\psi\rangle_{\max} = \frac{1}{\sqrt d}\sum_{i=1}^d e^{i\theta_i}|i\rangle$,
which is equally likely to occur across all possible paths.
In this case, the particleness of the pure state $|\psi\rangle_{\max} $ reaches the minimum.  Conversely, when the particleness in the interferometer attains the maximum, then the 
particle is in one of the basis states $\{|i\rangle\}$. In this case,
the waveness reaches the minimum.

By the analysis of elementary process in the interferometer \cite{Yadin2016}, we know the operations acting on a particle travelling through an interferometer are all strictly incoherent operations, which are characterized by the Kraus representations
\(\Lambda(\rho) = \sum_{\mu} K_{\mu} \rho K_{\mu}^{\dagger}\)
\cite{Baumgratz2014,YangD,Yadin2016}, where each strictly incoherent
Kraus operator \(K_\mu\) satisfies both
\(K_\mu \mathcal{I} K_\mu^\dagger \subset \mathcal{I}\) and
\(K_\mu^\dagger \mathcal{I} K_\mu \subset \mathcal{I}\) with $\mathcal{I}$ the set of incoherent state.
These conditions imply that \(K_\mu\) is compatible with the dephasing map \(\Delta(\cdot) = \sum_i |i\rangle\langle i|(\cdot)|i\rangle\langle i|\) in the sense that \(\Delta(K_\mu \rho K_\mu^\dagger)=
K_\mu \Delta(\rho) K_\mu^\dagger .\)

Next we are going to define a class of functions to quantify the waveness and particlenss in the $d$-path interferometer. We denote the state space of the particle by \(\mathcal H\), and the set of all quantum states on \(\mathcal H\) is denoted by \(D(\mathcal H)\). 
The  diagonal state $\rho_{\rm diag}$ with respect to the particle state $\rho$ is denoted by \(\rho_{\rm diag} = \Delta(\rho)\).
Consequently, we define a function $D: D(\mathcal{H})\times D(\mathcal{H})\rightarrow \mathbb{R}_{\geq 0}$, where $\mathbb{R}_{\geq 0}$ denotes the set of nonnegative real numbers, such that for any quantum states $\rho$ and $\sigma$, it has
\begin{enumerate}
    \item[(1)](Nonnegativity) $D(\rho,\sigma)\geq 0$ and $D(\rho,\sigma)=0$ if and only if $\rho=\sigma$.
    \item[(2)] (Monotonicity) $D$ does not increase for strictly incoherent operation $\Phi$, that is, $D(\Phi(\rho), \Phi(\sigma))\leq D(\rho,\sigma)$ .
    \item[(3)] (Joint convexity) $D(\sum_{i}p_{i}\rho_{i},\sum_{i}p_{i}\sigma_{i})\leq \sum_{i}p_{i}D(\rho_{i},\sigma_{i})$  for any sets of states $\{\rho_{i}\}$ and $\{\sigma_{i}\}$ and any probability distribution $\{p_i\}$. 
\end{enumerate}
The function $D$ may not necessarily be a metric. Next we show that such a function $D$ can be used to quantify the 
wave and particle behavior in a multi-path interferometer.

\begin{thm}\label{thm1}
    For any quantum state $\rho$, 
    \begin{equation}
        \mathcal{W}(\rho)=D(\rho,\rho_{\rm{diag}})
    \end{equation}
 is a reasonable measure of wave feature.
\end{thm}

\begin{proof}
Now we will prove  $\mathcal{W}(\rho)$ satisfies the conditions (W1)-(W4). The first item (W1) is obvious.

For any pure state $|\phi\rangle=\sum_{i=1}^{d} 
\sqrt{p_{i}}|i\rangle$, since $(\frac{1}{d},\frac{1}{d},\cdots,\frac{1}{d}) \prec (p_{1},p_{2},\cdots, p_{d})$,  there is a strictly incoherent operation $\Phi$ such that $\Phi(|\psi\rangle_{\max} \langle\psi|) = |\phi\rangle \langle\phi|$  \cite{YangD}. Because of the monotonicity of the function $D$ under strictly incoherent operations, it follows that
\begin{align}\label{Eq. SIO}
   D(|\phi\rangle \langle\phi|, (|\phi\rangle \langle\phi|)_{\text{diag}}) &= D(\Phi(|\psi\rangle_{\max} \langle\psi|), \Delta \circ \Phi(|\psi\rangle_{\max} \langle\psi|))\nonumber \\
   &= D(\Phi(|\psi\rangle_{\max} \langle\psi|), \Phi \circ \Delta(|\psi\rangle_{\max} \langle\psi|))\nonumber\\
   &\leq D(|\psi\rangle_{\max} \langle\psi|, (|\psi\rangle_{\max} \langle\psi|)_{\text{diag}}). 
\end{align}
Then for any mixed state $\rho=\sum_{i}p_{i}|\phi_{i}\rangle\langle \phi_{i}|$, by the joint convexity of the function $D$ and Eq. (\ref{Eq. SIO}), we have
\begin{eqnarray*}
    \mathcal{W}(\rho)&=&D(\sum_{i}p_{i}|\phi_{i}\rangle\langle \phi_{i}|,(\sum_{i}p_{i}|\phi_{i}\rangle\langle \phi_{i}|)_{\text{diag}})\\
    &\leq& \sum_{i}p_{i}D(|\phi_{i}\rangle\langle \phi_{i}|, (|\phi_{i}\rangle\langle \phi_{i}|)_{\text{diag}})\\
    &\leq& D(|\psi\rangle_{\max} \langle\psi|, (|\psi\rangle_{\max} \langle\psi|)_{\text{diag}}).
\end{eqnarray*}
Thus, $\mathcal{W}(\rho)$ reaches its global maximum if $\rho$ is pure with equal diagonal entries, i.e.  $\langle i|\rho|i\rangle = 1/d$ for all $i$.  

As to (W3), for any permutation matrix $P$, we have $\mathcal{W}(P\rho P^\dagger)=D(P\rho P^\dagger, (P\rho P^\dagger) _{\rm{diag}})=D(P\rho P^\dagger, P\rho_{\rm{diag}}P^\dagger)\leq \mathcal{W}(\rho)$ because the permutation $P$ is strictly incoherent. Meanwhile, applying the inverse operation  to quantum state $P\rho P^\dagger$, we obtain $\mathcal{W}(\rho)\leq W(P\rho P^\dagger)$.
 Therefore, $\mathcal{W}(\rho)$ is invariant under the permutations of the path labels.

As to (W4), for any quantum state $\rho$ and $\sigma$ and $0\leq\lambda\leq 1$, we have
\begin{align*}
    \mathcal{W}(\lambda \rho+(1-\lambda)\sigma)&=D(\lambda \rho+(1-\lambda)\sigma, \lambda \rho_{\text{diag}}+(1-\lambda)\sigma_{\text{diag}})\\
    &\leq \lambda D(\rho,\rho_{\text{diag}})+(1-\lambda)D(\sigma,\sigma_{\text{diag}})\\
    &=\lambda \mathcal{W}(\rho)+(1-\lambda)\mathcal{W}(\sigma),
\end{align*}
where the inequality is due to the joint convexity of $D$. This completes the proof.
\end{proof}

\begin{thm}\label{thm2}
    For any quantum state $\rho$, 
    \begin{equation}
        \mathcal{P}(\rho)=D(\rho_{\rm{diag}},\frac{I}{d})
    \end{equation}
is a reasonable measure of particle feature.
\end{thm}

\begin{proof}
The proofs for the items (P1), (P3) and (P4) are similar to that of (W1), (W3) and (W4) of Theorem \ref{thm1}, so we omit them here. Now we prove 
$\mathcal{P}(\rho)$ satisfies the condition (P2). 
For any quantum state $\rho=\sum_{i,j=1}^{d}\rho_{ij}|i\rangle\langle j|$, by definition we know $\mathcal{P}(\rho)$ is factually a function of the probability vector $\vec{\alpha}=(\rho_{11}, \rho_{22}, \cdots,\rho_{dd})^{T}$. Let $f(\vec{\alpha})=\mathcal{P}(\rho)$, by the condition (P3), we know $f(\vec{\alpha})$ is symmetric, i.e., $f(P\vec{\alpha})=f(\vec{\alpha})$ for any permutation matrix $P$.
Combining the convexity of $D$, it follows that $f$ is Schur-convex. By the majorization relation \cite{matrix}, $\vec{\alpha}\prec (1,0,\cdots,0)^{T}$,
we know $f(\vec{\alpha})\leq f((1,0,\cdots,0)^{T})$, which means that $\mathcal{P}(\rho)\leq \mathcal{P}(\sigma)$, where $\sigma$ is one of the basis states $\{|i\rangle\}$.  This completes the proof.
\end{proof}

{Refs. \cite{n3,du1} have constructed two classes of functions to quantify the waveness  based on coherence measures, and established the corresponding complementarity relations in the multi-path interferometer. In fact, the wave measure is not necessarily a coherence measure, such as the wave measure based on the $l_2$ norm \cite{GR,durr}, $l_p$ norm \cite{Tao},
fidelity \cite{xiong} . Therefore, Theorem 1 provides a more general method for constructing wave measure.
Ref. \cite{du1} also asks whether a coherence measure satisfying only non‑negativity, monotonicity, and convexity  can be able to characterize the wave–particle duality. We answer this question affirmatively. 
Furthermore, the monotonicity can be confined to the strictly incoherent operations instead of incoherent operations.} It extends the wave measures to a greater extent.


\section{The complementarity relation based on the trace distance}\label{sec III}

For any two density operators $\rho$ and $\sigma$, the trace distance between them is given by $T(\rho, \sigma) = \frac{1}{2} \|\rho - \sigma\|_{\text{tr}}$, where $\|\cdot\|_{\text{tr}}$ denotes the trace norm, defined as $\|A\|_{\text{tr}} = \text{Tr}\sqrt{A^\dagger A}$.  Since the trace distance is contractive under completely positive and trace-preserving maps and, together with its fundamental properties as a norm (non-negativity, positive homogeneity, and the triangle inequality), it satisfies all the requirements for the function $D$. So the trace distance is able to quantify waveness and particleness. 

Operationally speaking, according to the Holevo-Helstrom theorem, the waveness in terms of $T(\rho,\rho_{\text{diag}})$ upper bounds the advantage of any strategy that tries to discriminate $\rho$ and $\rho_{\text{diag}}$ under a one shot scenario. In practice, $\rho_{\text{diag}}$ can be prepared by making all the paths perfectly distinguishable, corresponding to thea maximally decoherence operation $\Delta$. The proposed method in \cite{Paul} for measuring coherence can be viewed as a special case of our approach where they only measure the difference of maximum intensity generated from $\rho$ and $\rho_{\text{diag}}$.

In this section, we study the complementarity relations in the $d$-path interferometer based on the trace distance. We adopt the normalized wave and particle measures to explore their restrictions from now on. So we need to analyze the maximum of the wave and particle measures first.

\begin{lem} \label{Eq. lem-im}
    If $\rho$ is pure state, then  $\rho-\rho_{\rm{diag}}$ has at most one positive eigenvalue.
\end{lem}
\begin{proof}
 By Weyl's inequality \cite{matrix}, for any $d\times d$ Hermitian matrices \(A\) and \(B\), we have \(\lambda_{i+j-1}^{\downarrow}(A+B) \le \lambda_i^{\downarrow}(A)+\lambda_j^{\downarrow}(B)\), where \(\lambda_k^{\downarrow}(\cdot)\) denotes the \(k\)-th largest eigenvalue of a Hermitian matrix, $i,j,k=1,2,\cdots, d$. Applying this with \(A=-\rho_{\rm diag}\), \(B=\rho\), and taking \(i=1\), \(j=2\), yields  
\[
\lambda_2^{\downarrow}(-\rho_{\rm diag}+\rho) \le \lambda_1^{\downarrow}(-\rho_{\rm diag}) + \lambda_2^{\downarrow}(\rho).
\]
Since \(\rho\) is a rank-one positive semidefinite matrix, its second largest eigenvalue vanishes, i.e. \(\lambda_2^{\downarrow}(\rho)=0\). Furthermore, \(-\rho_{\rm diag}\) is negative semidefinite, so \(\lambda_1^{\downarrow}(-\rho_{\rm diag})\le 0\). Consequently,
\[
\lambda_2^{\downarrow}(\rho-\rho_{\rm diag}) \le \lambda_1^{\downarrow}(-\rho_{\rm diag}) \le 0,
\]
which implies that \(\rho-\rho_{\rm diag}\) has at most one positive eigenvalue.   
\end{proof}


\begin{lem}\label{Eq. thm1}
For any  $d$-dimensional quantum state $\rho$, we have $||\rho-\rho_{\rm{diag}}||_{\rm{tr}}\leq \frac{2(d-1)}{d}$, {and the equality holds if and only if $\rho = |\psi\rangle_{\max}\langle\psi|$.} 
\end{lem}

\begin{proof}
Since the trace norm is a convex function, it suffices to prove the upper bound $||\rho-\rho_{\rm{diag}}||_{\rm{tr}}$ for pure states.
For any pure state $\rho = |\psi\rangle\langle\psi|$ with $|\psi\rangle = \sum_{i=1}^{d} c_i |i\rangle$ and $\sum_i |c_i|^2 = 1$,
the corresponding diagonal state is $\rho_{\text{diag}} = {\rm diag} (|c_1|^2, \dots, |c_d|^2)$.
We denote the matrix $A = \rho - \rho_{\text{diag}}$. By Lemma \ref{Eq. lem-im}, we know $A$ has at most one positive eigenvalue. First, if $A$ doesn't have positive eigenvalue, then $A$ must be a zero matrix due to ${\rm Tr}(A) = 0$. This is the trivial case. Second, if $A$ has one positive eigenvalue  $\lambda_+$, then the trace norm is $||A||_{\rm {tr}} = 2\lambda_+$.  In appendix \ref{app:C}, for any $\lambda_+>0$, it can be shown that
\begin{equation}\label{Eq. equ}
    \sum_{i=1}^{d} \frac{|c_i|^2}{\lambda_+ + |c_i|^2} = 1,
\end{equation}
and
\begin{equation}\label{Eq. upper}
  \sum_{i=1}^{d} \frac{|c_i|^2}{\lambda_+ + |c_i|^2}\leq \frac{1}{\lambda_+ +1/d}, 
\end{equation}
and the equality holds when $|c_i|^2 = \frac{1}{d}$ for $i=1,2,\cdots, d$. 
Eqs. (\ref{Eq. equ}) and (\ref{Eq. upper}) together 
implies $\lambda_+ \leq 1 - \frac{1}{d}$. Thus $||\rho-\rho_{\text{diag}}||_{\rm{tr}}=2\lambda_{+}\leq \frac{2(d-1)}{d}$, and the equality holds if and only if $|c_{i}|^{2}=\frac{1}{d}$ for all $i$ which means that  $\rho = |\psi\rangle_{\max}\langle\psi|$. This completes the proof.
\end{proof}

Lemma \ref{Eq. thm1} shows the maximum for the trace distance $||\rho-\rho_{\rm{diag}}||_{\rm{tr}}$, which will serve as a modification coefficient to quantify waveness. 

\begin{thm}
For any  $d$-dimensional quantum state $\rho$,  
\begin{equation}
    \mathcal{W_{\rm{tr}}}(\rho)=\frac{d}{2(d-1)}||\rho-\rho_{\rm{diag}}||_{\rm{tr}}
\end{equation}
is a reasonable measure of wave feature and $0\leq \mathcal{W_{\rm{tr}}}(\rho)\leq 1$.
\end{thm}

Now we come to the normalized particle measure.
Here we need a lemma for the modification coefficient, which is proved in Appendix \ref{app:D}.

\begin{lem} \label{Eq. vec}
    For any probability vector $\vec{x}=(x_1,x_2,\dots,x_d)$, we have $\sum_{i=1}^{d}|x_{i}-\frac{1}{d}|\leq \frac{2(d-1)}{d}$, the equality holds if and only if $x_{i}=1$ for some $i$.
\end{lem}

Based on Lemma \ref{Eq. vec}, we derive a normalized particle measure.

\begin{thm}
   For any  $d$-dimensional quantum state $\rho$,  
   \begin{equation}
       \mathcal{P_{\rm{tr}}}(\rho)=\frac{d}{2(d-1)}||\rho_{\rm{diag}}-\frac{I}{d}||_{\rm{tr}}
   \end{equation}
is a reasonable measure of particle feature and $0\leq \mathcal{P_{\rm{tr}}}(\rho)\leq 1$.
\end{thm}


\subsection{The wave-particle duality in low dimensional systems}

First, we study the wave-particle duality in low dimensional systems. 

\begin{thm}\label{pro-three}
  For any qubit and qutrit state $\rho$, we have the wave-particle duality 
     \begin{equation}\label{eq bound1}
         \mathcal{W_{\rm{tr}}}(\rho)^{2}+\mathcal{P_{\rm{tr}}}(\rho)^{2}\leq 1.
     \end{equation} 
\end{thm}

\begin{proof}
    For any qubit state $\rho$, we express it with Bloch representation
\[
\begin{gathered}
\rho=\frac{1}{2}(I+\vec{r}\cdot\vec{\sigma})=\frac{1}{2}
\begin{pmatrix}
1+r_{z} & r_{x}-ir_{y}  \\
r_{x}+ir_{y} & 1-r_{z} \\
\end{pmatrix}
\end{gathered},
\]
where $\vec{r}=(r_{x},r_{y},r_{z})\in \mathbb{R}^{3}$  and $\vec{\sigma}=(\sigma_x,\sigma_y,\sigma_z)$ with three Pauli operators, $\sigma_x=|0\rangle\langle 1| + |1\rangle\langle 0|$, $\sigma_y={\rm{i} }(|1\rangle\langle 0|-|0\rangle\langle 1|)$, $\sigma_z=|0\rangle\langle 0| - |1\rangle\langle 1|$.
By direct calculation, we have
$\mathcal{W_{\rm{tr}}}(\rho)^{2}=||\rho-\rho_{\rm{diag}}||_{\rm{tr}}^{2}=r_{x}^{2}+r_{y}^{2}$ and $\mathcal{P_{\rm{tr}}}(\rho)^{2}=||\rho_{\rm{diag}}-\frac{I}{d}||_{\rm{tr}}^{2}=r_{z}^{2}$, which gives rise to
 \begin{equation}\label{eq upper 1}
     \mathcal{W_{\rm{tr}}}(\rho)^{2}+\mathcal{P_{\rm{tr}}}(\rho)^{2}=|\vec{r}|^{2}\leq 1
 \end{equation}
with $|\vec{r}|^{2}=r_{x}^{2}+r_{y}^{2}+r_{z}^{2}$. The proof for qutrit case is more technical and we leave it  in Appendix \ref{app:E}. 
\end{proof}

For qubit systems, Eq. (\ref{eq bound1}) reduces to an equality if and only if the quantum state is pure, which is consistent with the case without a detector in Ref. \cite{Englert}. 
In qutrit systems, the equality in Eq. (\ref{eq bound1}) holds if and only if $\rho = |\psi\rangle\langle\psi|$ with $|\psi\rangle = \frac{1}{\sqrt{3}} \sum_i e^{{\rm i}\theta_i} |i\rangle$, {or $\rho$ is basis-state, $\rho = |i\rangle\langle i|$, $i=1,2,3$.}
Therefore the wave-particle duality in
Theorem \ref{pro-three} is tight. One may conjecture that this duality in Eq.  (\ref{eq bound1})  might hold in any finite dimension. However, it is not true.
By 2000 randomly generated density matrices,  we find that the range of coordinates $(\mathcal{P_{\rm{tr}}}, \mathcal{W_{\rm{tr}}})$ exceeds the quarter-circle region $\mathcal{P_{\rm{tr}}}^{2}+\mathcal{W_{\rm{tr}}}^{2}\leq 1$  at $d = 6$ (See the yellow triangles in subfigure (c) in Fig. \ref{fig WP}). 
In fact, Englert $et~al$. have already discussed the violation of the standard duality  in Ref. \cite{EK} previously. The reason for this phenomenon is that the shape of the border line that limiting the coordinates $(\mathcal{P_{\rm{tr}}}, \mathcal{W_{\rm{tr}}})$ depends on the particular choice for wave and particle measures.

\begin{figure}[htbp]
\centering

\begin{minipage}{0.30\textwidth}
    \centering
    \includegraphics[width=\linewidth]{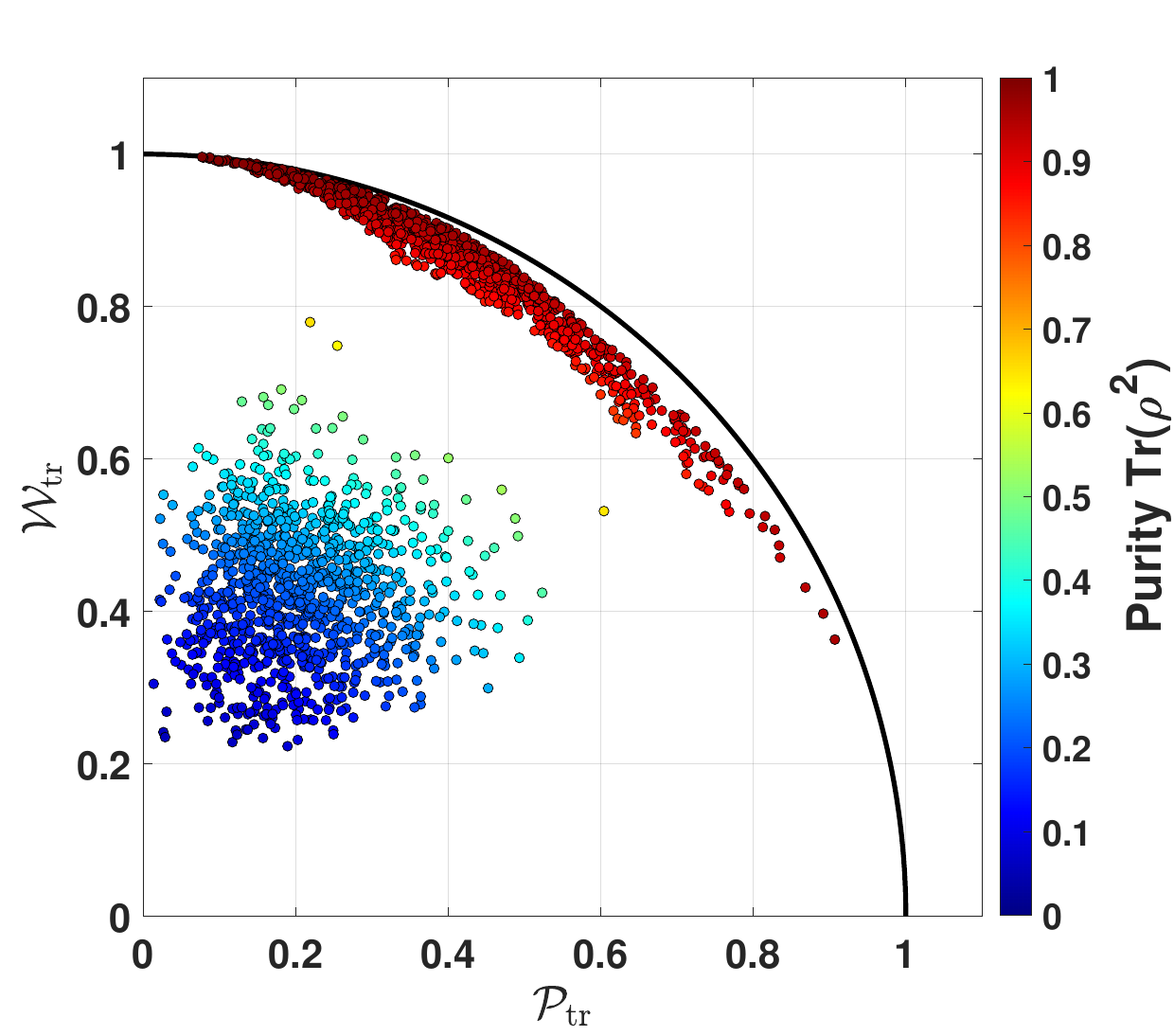}
    \\ (a)
\end{minipage}
\hfill
\begin{minipage}{0.30\textwidth}
    \centering
    \includegraphics[width=\linewidth]{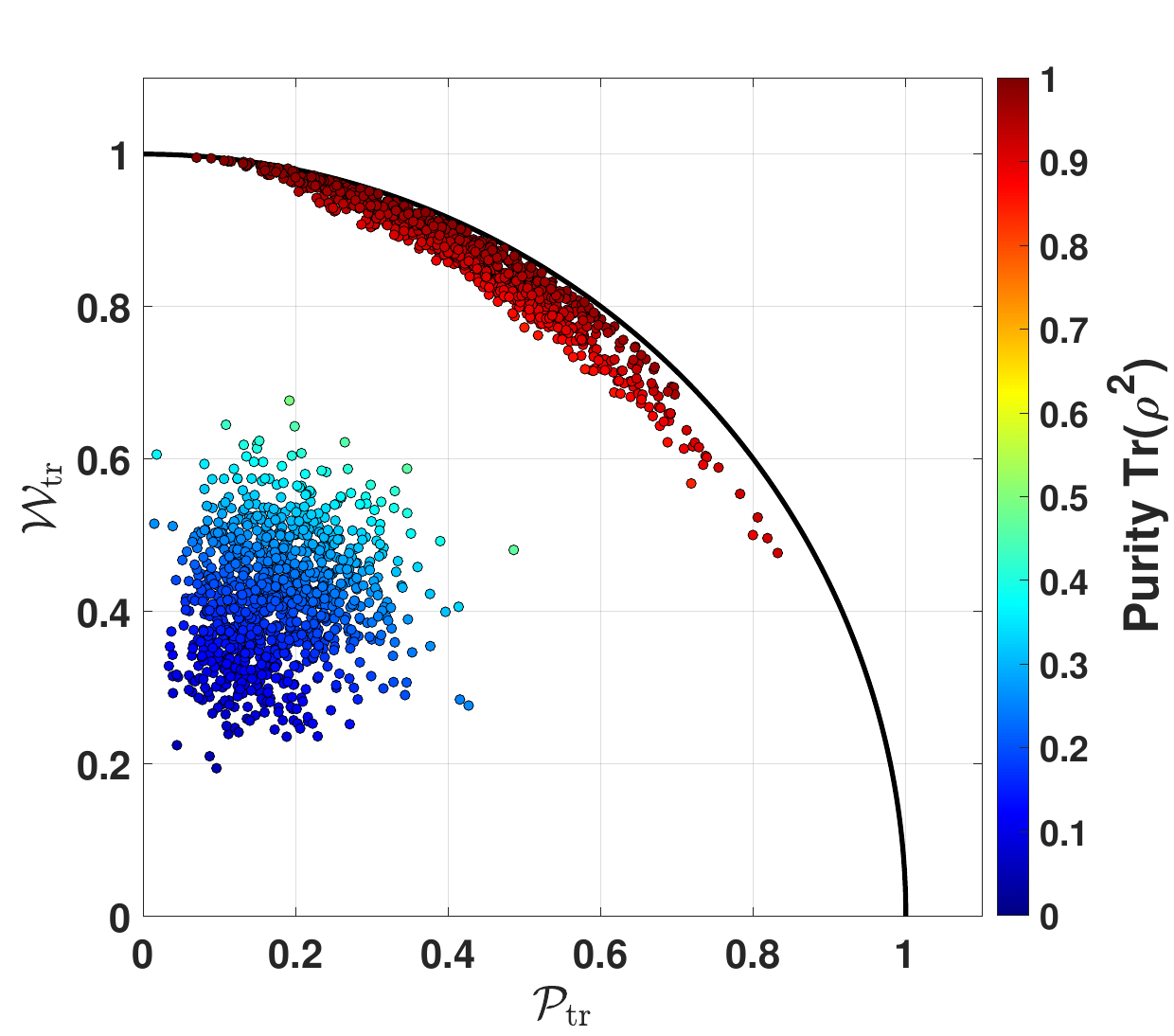}
    \\ (b)
\end{minipage}
\hfill
\begin{minipage}{0.30\textwidth}
    \centering
    \includegraphics[width=\linewidth]{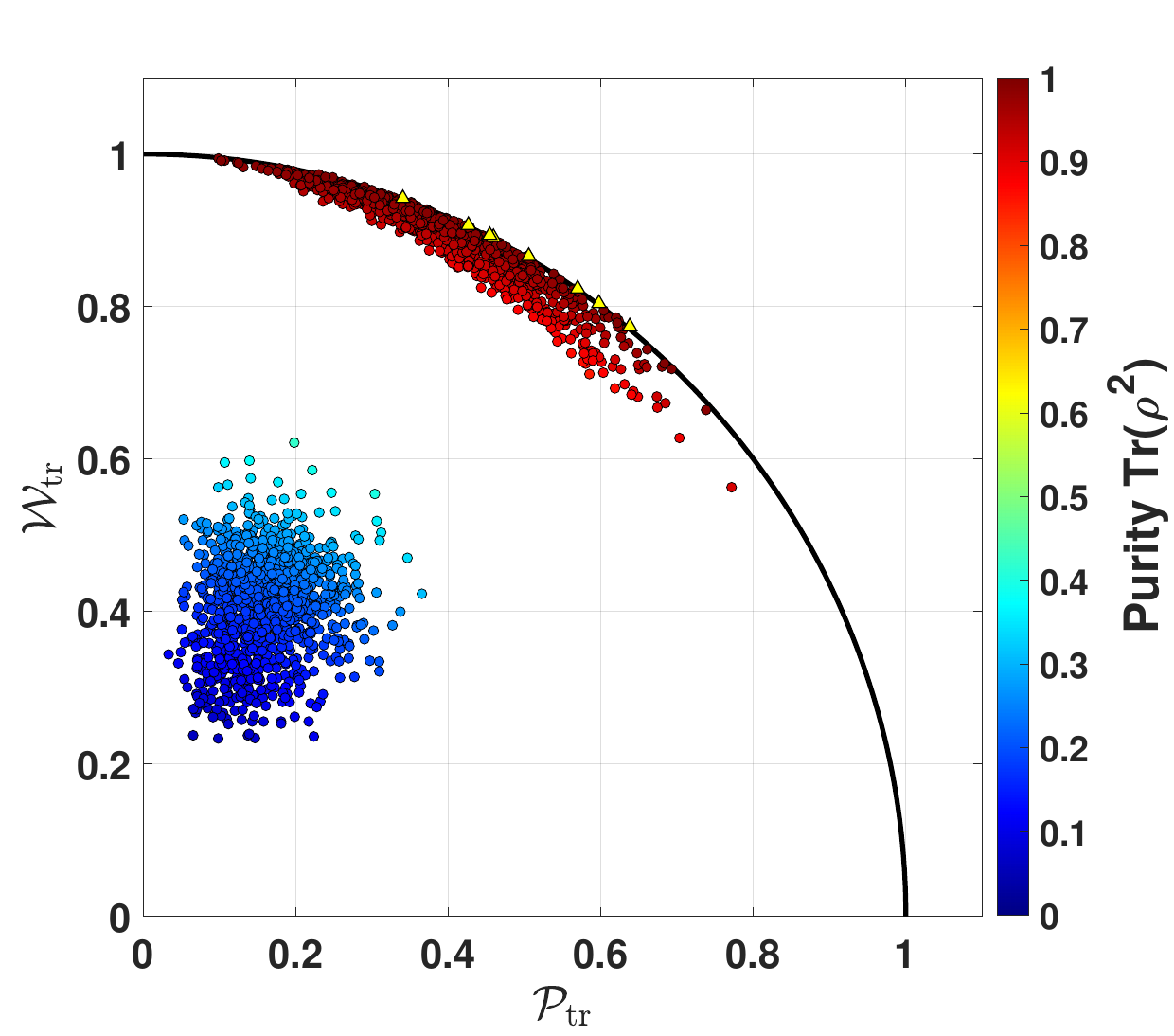}
    \\ (c)
\end{minipage}
\caption{(Color online)  Subfigures (a), (b), (c) display the $\mathcal{W_{\rm{tr}}}-\mathcal{P_{\rm{tr}}}$ diagrams for dimensions $d=4,5,6$, respectively. The black curve is the quarter-circle. The circular dots represent randomly generated quantum states, with their colors indicating the purity according to the color bar on the right. In subfigure (c), the yellow triangles represent the quantum states exceeding the black boundary. In another words, these quantum states violates the inquality $\mathcal{W_{\rm{tr}}}(\rho)^{2} + \mathcal{P_{\rm{tr}}}(\rho)^{2} \leq 1$. }
\label{fig WP}
\end{figure}

More specifically,  in a two-path interferometer,  Eq. (\ref{eq upper 1}) can be rephrased as a complete wave-particle-mixedness triality
  \begin{equation}\label{Eq. cor}
       \mathcal{W_{\rm{tr}}}(\rho)^{2}+\mathcal{P_{\rm{tr}}}(\rho)^{2} + 2( 1- {\rm{Tr}}(\rho^{2}))= 1,
    \end{equation}
the term $2( 1- {\rm{Tr}}(\rho^{2}))$ representing the mixedness provides the missing parts in Theorem \ref{pro-three} for qubit systems.

\subsection{The wave-particle duality in high dimensional systems}

Now we explore the  wave-particle duality  in high dimensional systems and we denote $r(A)$ the rank of matrix $A$.

\begin{thm}(Upper bound 1.)\label{Eq. thm-wp}
    For any $d$-dimensional quantum state $\rho$, we have wave-particle duality 
    \begin{equation}\label{Eq. wp}
 \mathcal{W_{\rm{tr}}}(\rho)^{2}+\mathcal{P_{\rm{tr}}}(\rho)^{2}\leq \frac{rd^{2}}{4(d-1)^{2}}({\rm{Tr}}(\rho^{2})-\frac{1}{d}),
    \end{equation}
 where $r=\max\{r_1,\ r_2\}$, $r_1=r(\rho-\rho_{\rm{diag}})$, and  $r_2= r(\rho_{\rm{diag}}-\frac{I}{d})$.
\end{thm}

\begin{proof}
According to the relation between the trace norm and Hilbert-Schmidt norm, $||A||_{\rm{tr}}\leq \sqrt{r(A)}||A||_{\rm{HS}}$, where  $||A||_{\rm{HS}}=({\rm {Tr}}(A^{\dag}A))^{\frac{1}{2}}$,  we derive that
   $||\rho-\rho_{\rm{diag}}||_{\rm{tr}}^{2}\leq r_{1}||\rho-\rho_{\rm{diag}}||_{\rm {HS}}^{2}$ and 
    $||\rho_{\rm{diag}}-\frac{I}{d}||_{\rm{tr}}^{2}\leq r_{2}||\rho_{\rm{diag}}-\frac{I}{d}||_{\rm {HS}}^{2}$. Subsequently, we have
    \begin{align*}
       \mathcal{W_{\rm{tr}}}(\rho)^{2}+\mathcal{P_{\rm{tr}}}(\rho)^{2}&= \frac{d^{2}}{4(d-1)^{2}}||\rho-\rho_{\rm{diag}}||_{\rm{tr}}^{2}+\frac{d^{2}}{4(d-1)^{2}}||\rho_{\rm{diag}}-\frac{I}{d}||_{\rm{tr}}^{2}\\
       &\leq \max\{r_{1},r_2\} \frac{d^{2}}{4(d-1)^{2}}(||\rho-\rho_{\rm{diag}}||_{\rm {HS}}^{2}+||\rho_{\rm{diag}}-\frac{I}{d}||_{\rm {HS}}^{2})\\
       &=\frac{rd^{2}}{4(d-1)^{2}}({\rm{Tr}}(\rho^{2})-\frac{1}{d}).
    \end{align*}
\end{proof}

For a two-path interferometer, the parameter \(r\) admits only two possible values, namely \(0\) and \(2\). When \(r=0\), this is equivalent to \(\rho = I/2\), under which  both \(\mathcal{P}_{\mathrm{tr}}(\rho)\) and \(\mathcal{W}_{\mathrm{tr}}(\rho)\) reach their minimum \(0\). This is the trivial case. When \(r=2\), the inequality in  Eq. (\ref{Eq. wp}) yields to equality in Eq. (\ref{Eq. cor}). Therefore, the wave-particle duality in Theorem \ref{Eq. thm-wp} is tight for the two-path interferometer.

\begin{thm}(Upper bound 2.)\label{thm-rank}
    For any $d$-dimensional quantum state $\rho$, we have wave-particle duality 
    \begin{equation}\label{}
 \mathcal{W_{\rm{tr}}}(\rho)^{2}+\mathcal{P_{\rm{tr}}}(\rho)^{2}\leq \frac{R^{'}d^{2}}{(d-1)^{2}}({\rm{Tr}}(\rho^{2})-\frac{1}{d}),
    \end{equation}
 where $R^{'}=\max\{R_{1},R_{2}\}$ with  $R_{1}=\frac{r(\rho)r(\rho_{\rm{diag}})}{r(\rho)+r(\rho_{\rm{diag}})}$ and $R_{2}=\frac{r(\rho_{\rm{diag}})d}{r(\rho_{\rm{diag}})+d}$.
\end{thm}

\begin{proof}
For any two quantum states $\rho$ and $\sigma$, 
based on the relation between trace norm and Hilbert-Schmidt norm given in Ref. \cite{strong}, 
we have $||\rho-\sigma||_{\rm{tr}}^{2}\leq 4R||\rho-\sigma||_{\rm{HS}}^{2}$, where $R=\frac{r(\rho)r(\sigma)}{r(\rho)+r(\sigma)}$. Thus,
\begin{align*}
       \mathcal{W_{\rm{tr}}}(\rho)^{2}+\mathcal{P_{\rm{tr}}}(\rho)^{2}&= \frac{d^{2}}{4(d-1)^{2}}||\rho-\rho_{\rm{diag}}||_{\rm{tr}}^{2}+\frac{d^{2}}{4(d-1)^{2}}||\rho_{\rm{diag}}-\frac{I}{d}||_{\rm{tr}}^{2}\\
       &\leq \max\{R_{1},R_2\} \frac{d^{2}}{(d-1)^{2}}(||\rho-\rho_{\rm{diag}}||_{\rm {HS}}^{2}+||\rho_{\rm{diag}}-\frac{I}{d}||_{\rm {HS}}^{2})\\
       &=\frac{R^{'}d^{2}}{(d-1)^{2}}({\rm{Tr}}(\rho^{2})-\frac{1}{d}).
    \end{align*}
\end{proof}

Up to now, we have derived two wave-particle dualities by utilizing the connection between the trace norm and the Hilbert-Schmidt norm. Next we present the third wave-particle duality based on the relation between trace norm and linear entropy.

\begin{thm}(Upper bound 3.)\label{wp-sl}
     For any $d$-dimensional quantum state $\rho$, we have wave-particle duality 
    \begin{equation}
 \mathcal{W_{\rm{tr}}}(\rho)^{2}+\mathcal{P_{\rm{tr}}}(\rho)^{2}\leq \frac{d^{2}}{(d-1)^{2}}[2-\frac{1}{d}-{\rm{Tr}}(\rho_{\rm{diag}}^{2})].
    \end{equation}
\end{thm}

\begin{proof}
For any two quantum states $\rho$ and $\sigma$, it has $||\rho-\sigma||_{\rm{tr}}^{2}\leq 2(||\rho-\sigma||_{\rm{HS}}^{2}+S_{L}(\rho)+S_{L}(\sigma))$, where  $S_{L}(\rho)=1-{\rm{Tr}(\rho^{2})}$  is the linear entropy \cite{strong}. Thus,  
  \begin{align*}
       \mathcal{W_{\rm{tr}}}(\rho)^{2}+\mathcal{P_{\rm{tr}}}(\rho)^{2}&= \frac{d^{2}}{4(d-1)^{2}}||\rho-\rho_{\rm{diag}}||_{\rm{tr}}^{2}+\frac{d^{2}}{4(d-1)^{2}}||\rho_{\rm{diag}}-\frac{I}{d}||_{\rm{tr}}^{2}\\
       &\leq  \frac{d^{2}}{2(d-1)^{2}}(||\rho-\rho_{\rm{diag}}||_{\rm {HS}}^{2}+||\rho_{\rm{diag}}-\frac{I}{d}||_{\rm {HS}}^{2}+S_{L}(\rho)+2S_{L}({\rho_{\rm{diag}}})+S_{L}(\frac{I}{d}))\\
       &=\frac{d^{2}}{(d-1)^{2}}[2-\frac{1}{d}-{\rm{Tr}}(\rho_{\rm{diag}}^{2})].
    \end{align*}
\end{proof}

Here we show the fourth wave-particle duality based on the relation between the trace norm and the relative entropy.

 \begin{thm}(Upper bound 4.)\label{wp-S}
      For any $d$-dimensional quantum state $\rho$, we  have wave-particle duality 
    \begin{equation}
 \mathcal{W_{\rm{tr}}}(\rho)^{2}+\mathcal{P_{\rm{tr}}}(\rho)^{2}\leq \frac{\ln(2)d^{2}}{2(d-1)^{2}}(\log{d}-S(\rho)).
 \end{equation}
 \end{thm}

\begin{proof}
According to the quantum Pinsker inequality  $S(\rho||\sigma)\geq \frac{1}{2\ln(2)}||\rho-\sigma||_{\rm{tr}}^{2}$, we have  
  \begin{align*}
       \mathcal{W_{\rm{tr}}}(\rho)^{2}+\mathcal{P_{\rm{tr}}}(\rho)^{2}&= \frac{d^{2}}{4(d-1)^{2}}||\rho-\rho_{\rm{diag}}||_{\rm{tr}}^{2}+\frac{d^{2}}{4(d-1)^{2}}||\rho_{\rm{diag}}-\frac{I}{d}||_{\rm{tr}}^{2}\\
       &\leq  \frac{\ln(2)d^{2}}{2(d-1)^{2}}(S(\rho||\rho_{\text{diag}})+S(\rho_{\text{diag}}||\frac{I}{d}))\\
       &=\frac{\ln(2)d^{2}}{2(d-1)^{2}}(\log{d}-S(\rho)).
    \end{align*}
\end{proof}

 Theorems \ref{Eq. thm-wp}, \ref{thm-rank}, \ref{wp-sl}, and \ref{wp-S} provide four upper bounds for the wave–particle duality from the perspectives of the rank, purity, and entropy.
 {For them, the maximally mixed state \(\rho = I/d\) trivially saturates the inequalities in which  both waveness and particleness vanish simultaneously.}
Now we analyze these upper bounds on restricting the wave and particle
behaviors.
   {\rm Consider the $d$-dimensional quantum state 
    \begin{equation}\label{eq ex state}
        \rho = (1-\varepsilon)|0 \rangle \langle 0| + \varepsilon \sigma,  
    \end{equation}
    where $\sigma$ is a mixed state and $\varepsilon \in [0, 1]$. 
Apparently, increasing  $\varepsilon$ results in  a higher degree of mixing of $\rho$. For $d=2,3,4,5,6,7$, we generate $50$ random density matrices by sampling different values of $\varepsilon$ uniformly from the interval $[0,1]$. 
The behaviors of these upper bounds in Theorems \ref{Eq. thm-wp}, \ref{thm-rank}, \ref{wp-sl}, and \ref{wp-S} with respect to $\varepsilon$ are plotted in Fig. \ref{fg3}.
    
For $d=2$ (See subfigure (a)), the upper bound of Theorem \ref{Eq. thm-wp} coincides with the exact value $\mathcal{W_{\rm{tr}}}(\rho)^{2}+\mathcal{P_{\rm{tr}}}(\rho)^{2}$. For $\varepsilon=0$, the upper bound of Theorem \ref{thm-rank} is better than that of Theorem \ref{wp-S}, whereas for $0<\varepsilon \leq 1$, the upper bound of Theorem \ref{wp-S} becomes tighter than that of Theorem \ref{thm-rank}. 

For $d=3$ (See subfigure (b)), the upper bound in Theorem \ref{Eq. thm-wp} is the best for $\varepsilon<0.098236$ and $\varepsilon >0.707827$, whereas the upper bound in Theorem \ref{wp-S} is the best for $0.098236 \leq \varepsilon \leq 0.707827$. For $\varepsilon=0$, the upper bound of Theorem \ref{thm-rank} coincides with the upper bound of Theorem \ref{Eq. thm-wp}. Moreover, the upper bound in Theorem \ref{wp-sl} is better than that in Theorem \ref{thm-rank} when $0<\varepsilon \leq 0.104738$. 

For $d=4$ (See subfigure (c)),   the upper bound in Theorem \ref{wp-S} is the best one for $0 < \varepsilon\leq 0.943534$, while that  in Theorem \ref{Eq. thm-wp} is best for $\varepsilon \geq 0.943534$. For $\varepsilon=0$, the upper bound in Theorem \ref{thm-rank}  is the closest to the exact value and the upper bound in Theorem \ref{Eq. thm-wp} coincides with that in Theorem \ref{wp-sl}. Moreover, the upper bound in Theorem \ref{wp-sl} is better than that in Theorem \ref{thm-rank} when $0<\varepsilon \leq 0.186393$. 

Now we observe and analyze the behaviors in $d=5,6,7$ (See subfigures (d), (e), (f)). For $\varepsilon=0$, the upper bound in Theorem \ref{thm-rank}  is the closest to the exact value for $d=5, 6, 7$. The upper bound of Theorem \ref{wp-S} remains the best when $\varepsilon \in (0, 0.497688]$ for $d=5$ and this range of $\varepsilon$ expands as the dimension $d$ increases; however, the upper bound in Theorem \ref{Eq. thm-wp} is the best when $\varepsilon \in (0.497688,1]$ for $d=5$ and this range of $\varepsilon$ shrinks as the dimension $d$ increases. Meanwhile, the upper bound in Theorem \ref{wp-sl} is better than that in Theorem \ref{Eq. thm-wp} when $\varepsilon \in [0, 0.052661)$ for $d=5$, and this range of $\varepsilon$ expands as the dimension $d$ increases. The upper bound in Theorem \ref{wp-sl} is better than that in Theorem \ref{thm-rank} when $\varepsilon \in [0, 0.228316)$ for $d=5$, and this range of $\varepsilon$ expands as the dimension $d$ increases. For $d=6$,  the bound of Theorem \ref{wp-sl} becomes the tightest when $\varepsilon \in (0, 0.014759)$ and this range of $\varepsilon$ expands as the dimension $d$ increases.

{The numerical analysis shows that the tightness of Theorems \ref{Eq. thm-wp}–\ref{wp-S} varies with \(d\) and \(\varepsilon\) for the quantum state in Eq. (\ref{eq ex state}). The lower bound in Theorem \ref{Eq. thm-wp} approaches the exact  value for pure or highly mixed states and coincide with it at \(d=2\). The lower bound in Theorem \ref{thm-rank} is near the exact  value at \(\varepsilon=0\) and  it is weaker elsewhere.  The lower bound in Theorem \ref{wp-sl} outperforms \ref{Eq. thm-wp} or \ref{thm-rank} in some low-to-mid mixing regions which grows with \(d\). The lower bound in  Theorem \ref{wp-S} is close to the exact  value for intermediate \(\varepsilon\), and its favorable range expands with \(d\). Thus, the optimal bound should be selected adaptively based on both dimension and mixing degree.}
     
 \begin{figure}[htbp]
\centering

\begin{minipage}{0.30\textwidth}
    \centering
    \includegraphics[width=\linewidth]{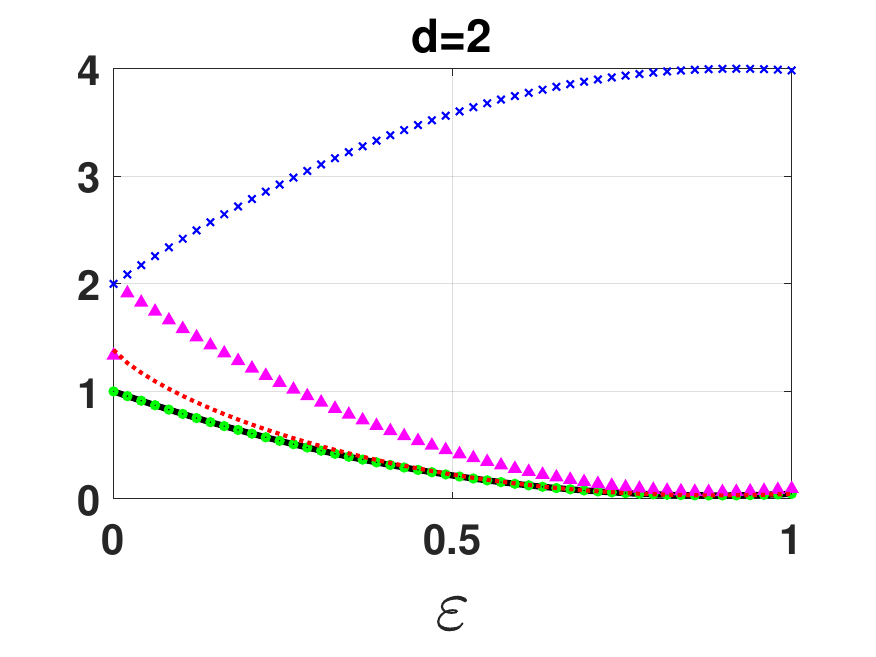}
    \\ (a)
\end{minipage}
\hfill
\begin{minipage}{0.30\textwidth}
    \centering
    \includegraphics[width=\linewidth]{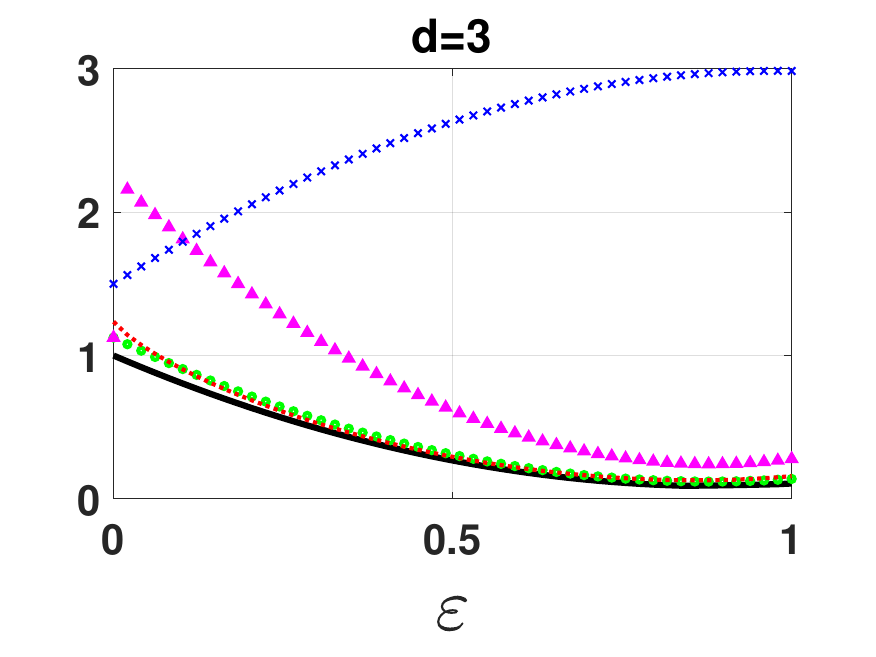}
    \\ (b)
\end{minipage}
\hfill
\begin{minipage}{0.30\textwidth}
    \centering
    \includegraphics[width=\linewidth]{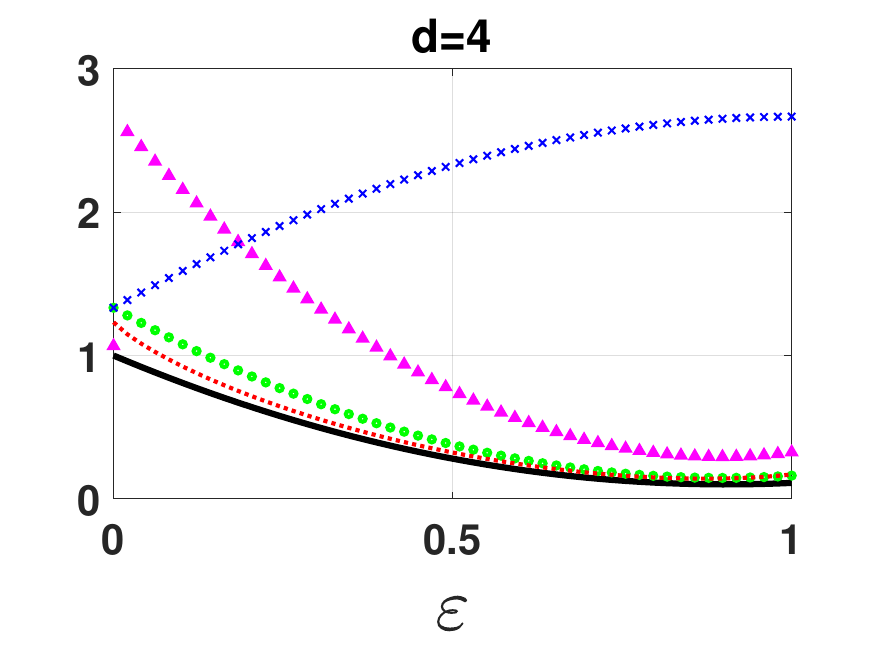}
    \\ (c)
\end{minipage}
\hfill
\begin{minipage}{0.30\textwidth}
    \centering
    \includegraphics[width=\linewidth]{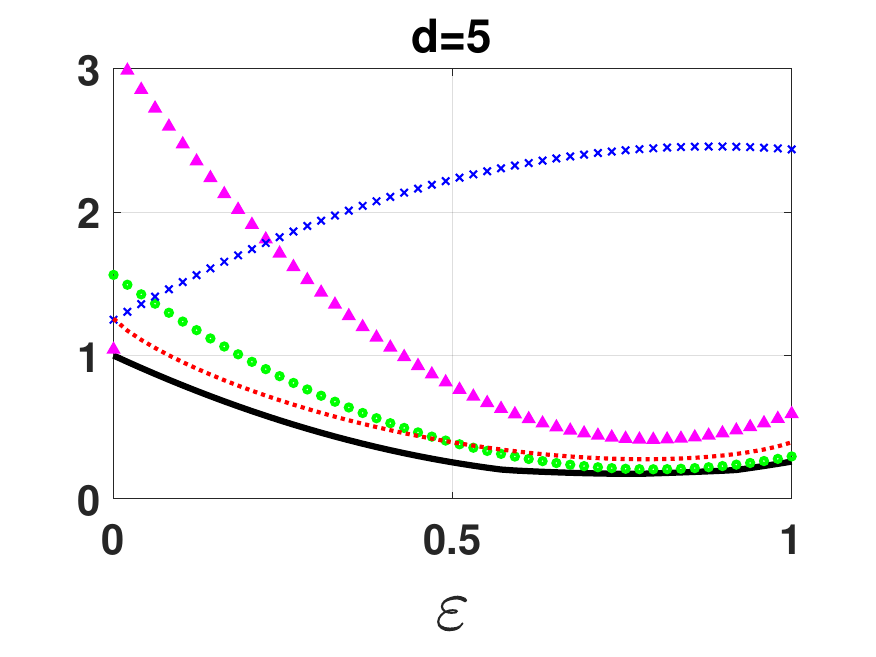}
    \\ (d)
\end{minipage}
\hfill
\begin{minipage}{0.30\textwidth}
    \centering
    \includegraphics[width=\linewidth]{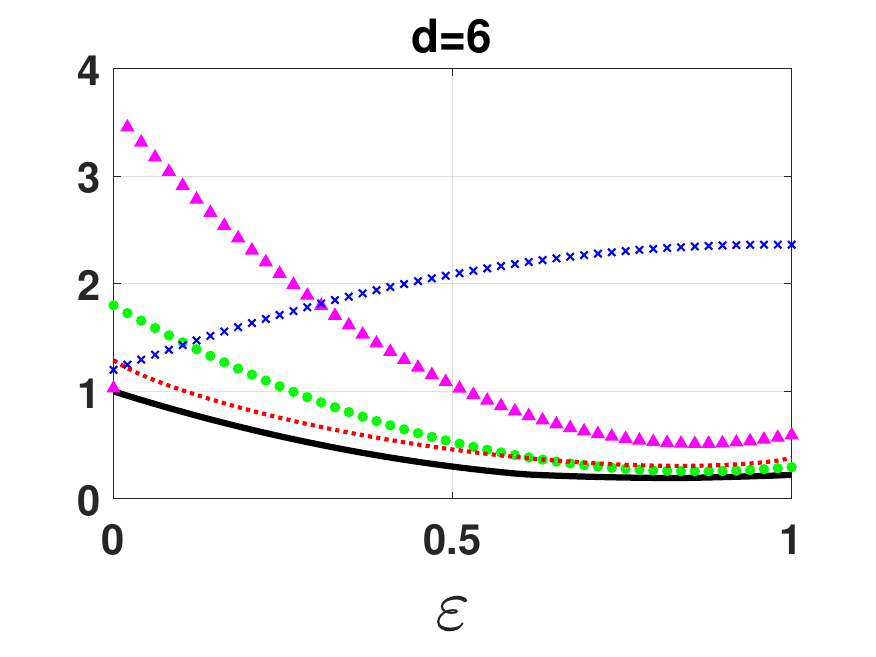}
    \\ (e)
\end{minipage}
\hfill
\begin{minipage}{0.30\textwidth}
    \centering
    \includegraphics[width=\linewidth]{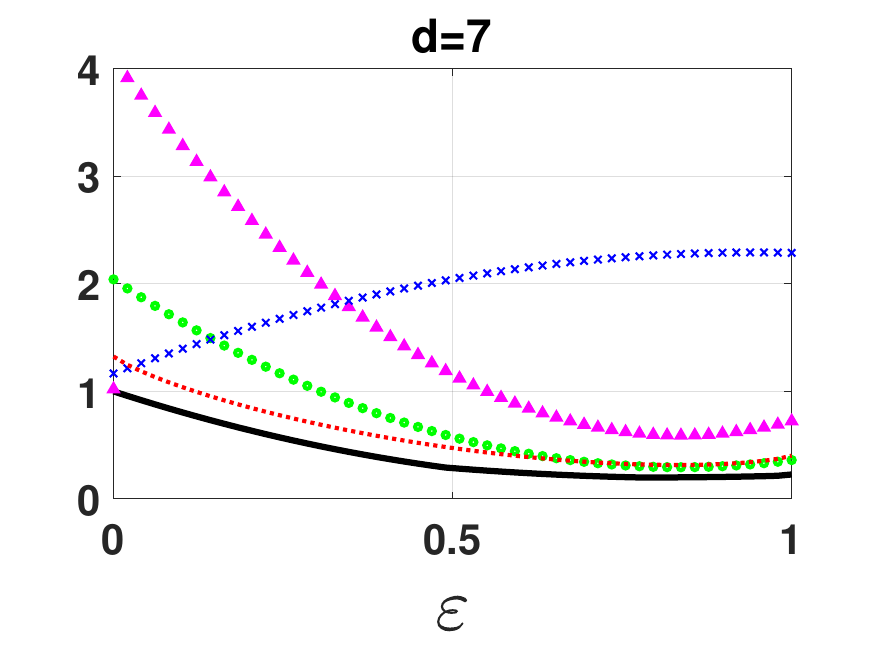}
    \\ (f)
\end{minipage}

\caption{(Color online)  
The illustration of the wave-particle dualities in Theorems \ref{Eq. thm-wp}, \ref{thm-rank}, \ref{wp-sl}, and \ref{wp-S} with respect to $\varepsilon$ for quantum states in Eq. (\ref{eq ex state}).
The black solid curve, green dotted curve, pink triangular curve, blue cross-shaped curve,  red dashed curve represent the 
exact value of $\mathcal{W_{\rm{tr}}}(\rho)^{2}+\mathcal{P_{\rm{tr}}}(\rho)^{2}$ and the upper bounds in Theorems \ref{Eq. thm-wp}, \ref{thm-rank}, \ref{wp-sl} and \ref{wp-S} respectively.
}
\label{fg3}
\end{figure}

\subsection{The  triality relation in a multi-path interferometer}\label{sub-C}

In this section, we mainly consider the complementarity relation among wave, particle, and mixedness. Similarly to the wave and particle measures in Theorems \ref{thm1} and \ref{thm2}, for any function $D$ defined in section II, for any quantum state $\rho$, we define
\begin{equation}
    \mathcal{M}(\rho)=C-D(\rho,\frac{I}{d}), \quad C=\max_\rho D(\rho,\frac{I}{d}).
\end{equation}
Then $\mathcal{M}(\rho)$ can be viewed as a quantifier of mixedness, with $0\leq \mathcal{M}(\rho)\leq C$. If we additionally require the function $D$ to satisfy the triangle inequality, then a natural linear complementarity relation for the waveness and particleness in Theorems \ref{thm1} and \ref{thm2} is
\begin{equation}\label{eq lower 1}
    \mathcal{W}(\rho) + \mathcal{P}(\rho) +\mathcal{M} (\rho)\geq C.
\end{equation}
Furthermore, by the mean inequality $a^{2}+b^{2}\geq \frac{1}{2}(a+b)^{2}$, we also obtain a quadratic complementarity relation
\begin{equation}\label{eq lower 2}
    \mathcal{W}^2(\rho) + \mathcal{P}^2(\rho)  \geq \frac{1}{2}(C-\mathcal{M} (\rho))^2.
\end{equation}
The complementarity relations in Eqs. (\ref{eq lower 1}) and (\ref{eq lower 2}) imply that the waveness, particleness, and mixedness cannot all be small simultaneously.

In fact, if we specify the function $D$ to be the relative entropy, the complementarity relation for wave-particle-mixedness in Eq. (\ref{eq lower 1}) reduces to the equality
\begin{equation}
     \mathcal{W}(\rho) +\mathcal{P}(\rho)+ \mathcal{M}(\rho)=\log{d},
\end{equation}
where $\mathcal{W}(\rho)=S(\rho||\rho_{\text{diag}})$ is the waveness, $\mathcal{P}(\rho)=S(\rho_{\text{diag}}||\frac{I}{d})$ is the particleness, and $\mathcal{M}(\rho)=\log{d}-S(\rho||\frac{I}{d})$ is the mixedness, with $S(\rho||\sigma) = \mathrm{tr}(\rho \log \rho) - \mathrm{tr}(\rho \log \sigma)$.

Next, we adopt the linear entropy as a measure of mixedness, that is, $M_1(\rho)=1-\text{Tr}{\rho}^{2}$. By transferring the purity term ${\rm Tr}({\rho}^2)$ to the left hand side of the inequalities in Theorem \ref{Eq. thm-wp} and \ref{thm-rank}, we derive two wave-particle-mixedness triality relations.

\begin{proposition}(Triality 1.)  \label{th d wpe}
     For any $d$-dimensional quantum state $\rho$, we have the wave, particle and mixedness satisfy
  \begin{equation}  \label{eq wpe}
    \mathcal{W_{\rm{tr}}}(\rho)^{2}+\mathcal{P_{\rm{tr}}}(\rho)^{2}+\frac{{r}d^2}{4(d-1)^2} M_1(\rho) \leq \frac{rd}{4(d-1)},
  \end{equation}
 where $r=\max\{r_1,\ r_2\}$, $r_1=r(\rho-\rho_{\rm{diag}})$, and  $r_2= r(\rho_{\rm{diag}}-\frac{I}{d})$.
\end{proposition}

\begin{proposition}(Triality 2.)
   For any $d$-dimensional quantum state $\rho$, we also have the wave, particle and mixedness satisfy
   \begin{equation}
      \mathcal{W_{\rm{tr}}}(\rho)^{2}+\mathcal{P_{\rm{tr}}}(\rho)^{2}+ \frac{R^{'}d^{2}}{(d-1)^{2}} M_1(\rho) \leq \frac{R^{'}d}{d-1},  
   \end{equation}
      where $R^{'}=\max\{R_{1},R_{2}\}$ with  $R_{1}=\frac{r(\rho)r(\rho_{\rm{diag}})}{r(\rho)+r(\rho_{\rm{diag}})}$ and $R_{2}=\frac{r(\rho_{\rm{diag}})d}{r(\rho_{\rm{diag}})+d}$.
\end{proposition}

Now we utilize  von Neumann entropy  as mixedness, that is, $M_2(\rho)=S(\rho )$.
Similarly, by Theorem \ref{wp-S}, we can derive the third wave-particle-mixedness triality.
\begin{proposition}(Triality 3.)
   For any $d$-dimensional quantum state $\rho$, we have the wave, particle and mixedenss satisfy
   \begin{equation}
        \mathcal{W_{\rm{tr}}}(\rho)^{2}+\mathcal{P_{\rm{tr}}}(\rho)^{2}+ \frac{\ln(2)d^{2}}{2(d-1)^{2}}M_2(\rho)\leq \frac{\ln(2)d^{2}}{2(d-1)^{2}}\log{d}.
   \end{equation}
\end{proposition}

The three wave-particle-mixedness trialities are derived 
directly from the corresponding wave-particle dualities respectively.  In the following, we will present a complete complementarity relation. Frist we define
\begin{equation}
    M(\rho)=\frac{rd}{4(d-1)}-\frac{d^{2}}{4(d-1)^{2}}||\rho-\rho_{\rm{diag}}||_{\rm{tr}}^{2}-\frac{d^{2}}{4(d-1)^{2}}||\rho_{\rm{diag}}-\frac{I}{d}||_{\rm{tr}}^{2}.
\end{equation}
By the convexity of the trace norm, one can verify that 
$M(\rho)$ fulfills  the following properties:
\begin{enumerate}
    \item[(1)]$M(\rho)$ reaches its global minimum if the state $\rho$ is a pure state.
    \item[(2)] $M(\rho)$ reaches its global maximum if the state $\rho$
is a maximum mixed state.
    \item[(3)] $M(\rho)$ is concave.
\end{enumerate}
Note that $M(\rho)$ exhibits some properties of mixedness.  Based on $M(\rho)$, a generalized triality relation is given below. 
\begin{proposition}(Triality 4.)
     For any quantum state $\rho$, we have the triality relation
     \begin{equation}
         \mathcal{W_{\rm{tr}}}(\rho)^{2}+\mathcal{P_{\rm{tr}}}(\rho)^{2}+M(\rho)= \frac{rd}{4(d-1)},
     \end{equation}
where $r=\max\{r_1,\ r_2\}$, $r_1=r(\rho-\rho_{\rm{diag}})$, and  $r_2= r(\rho_{\rm{diag}}-\frac{I}{d})$.
\end{proposition}


\section{Conclusions}\label{sec IV}

In summary, we have developed a  method to quantify
both wave and particle behaviors in a multi-path interferometer. In particular, we find that the trace distance is a good candidate for wave and particle measures. As a result, some complementarity relations including the wave-particle duality, wave-particle-mixedness triality, and generalized triality relation are established respectively.

As we know, the trace distance is not a well-defined coherence measure, so this work provides  the evidence for the difference between the waveness and coherence further. Since the difference between waveness and coherence is still unclear, so it is challenging to explore the nature of the waveness. Additionally, the restrictions of waveness and particleness is also an interesting topic. For example, for the wave and particle measure in terms of the trace distance, we have proved $ \mathcal{W_{\rm{tr}}}(\rho)^{2}+\mathcal{P_{\rm{tr}}}(\rho)^{2}\leq 1$ for qubit and qutrit systems and find the violation for $d=6$. But we don't know whether the inequality is true for $d=4$ or $d=5$. 


\section{Acknowledgment}
{We thank the anonymous referees for the valuable comments. We are also grateful to Teng Ma and Xiongfeng Ma for helpful discussions. M. J. Zhao thanks the center for Quantum Information,
Institute for Interdisciplinary Information Sciences,
Tsinghua University for hospitality. This work is partially supported by the National Natural Science Foundation of China (Nos. 11671201 and 12171044) and Postgraduate Research \& Practice Innovation Program of Jiangsu Province (KYCX25\_0627).}

\section{Appendix}





\appendixsection{The proof of Eq. (\ref{Eq. equ}) and Eq. (\ref{Eq. upper})}\label{app:C}

First, for any pure state $\rho=|\psi\rangle\langle\psi|$ and $|\psi\rangle = \sum_{i=1}^{d} c_i |i\rangle$ with $\sum |c_i|^2 = 1$, we will  prove  all positive eigenvalues of matrix $\rho-\rho_{\rm{diag}}$ satisfy 
\begin{equation}
    \sum_{i=1}^{d} \frac{|c_i|^2}{\lambda + |c_i|^2}=1.
\end{equation}

\begin{proof}
Let $\lambda$ be the positive eigenvalue of $A=|\psi\rangle\langle\psi|-D$ with $D = \rho_{\text{diag}} = {\rm diag}(|c_1|^2, \dots, |c_d|^2))$ and $|v\rangle$ is the corresponding eigenvector, that is $A|v\rangle = \lambda|v\rangle$, i.e.
\begin{equation}\label{Eq. e1}
  (|\psi\rangle\langle\psi| - D)|v\rangle = \lambda|v\rangle. 
\end{equation}
Set $\langle\psi|v\rangle = \alpha$, the Eq. (\ref{Eq. e1}) becomes
\begin{equation}\label{Eq. e2}
    \alpha|\psi\rangle - D|v\rangle = \lambda|v\rangle,
\end{equation}
which gives $ (D + \lambda I)|v\rangle = \alpha|\psi\rangle$. Since $\lambda > 0$,  $(D + \lambda I)$ is invertible, then
\begin{equation}\label{eq app-1}
   |v\rangle = \alpha (D + \lambda I)^{-1}|\psi\rangle, 
\end{equation}
where $(D + \lambda I)^{-1}$ is a diagonal matrix with diagonal elements $1/(|c_i|^2 + \lambda)$.
Since $\langle\psi|v\rangle = \alpha$, substitute into Eq. (\ref{eq app-1})
\begin{equation}\label{eq app-2}
    \langle\psi|\alpha (D + \lambda I)^{-1}|\psi\rangle = \alpha,
\end{equation}
Note that $\alpha \neq 0$. Or else, combining Eq. (\ref{Eq. e1}) and Eq. (\ref{Eq. e2}), we have $A|v\rangle=-D|v\rangle=\lambda|v\rangle$, which implies $-\lambda$ is the eigenvalue of $D$. This contradicts the nonnegative eigenvalue $|c_{i}|^{2}$ of $D$. 
Then, dividing both sides of Eq. (\ref{eq app-2}) by $\alpha$ gives
\begin{equation}\label{Eq. e3}
    \langle\psi|(D + \lambda I)^{-1}|\psi\rangle = 1.
\end{equation}
Computing the left-hand side of Eq. (\ref{Eq. e3}), we get
\begin{equation*}
   \langle\psi|(D + \lambda I)^{-1}|\psi\rangle 
   = \sum_{i=1}^{d} \frac{|c_i|^2}{|c_i|^2 + \lambda}. 
\end{equation*}
Thus, the any positive eigenvalue $\lambda$ of $A$ satisfies
\begin{equation*}
    \sum_{i=1}^{d} \frac{|c_i|^2}{|c_i|^2 + \lambda} = 1.
\end{equation*}
\end{proof}
Then, for any $\lambda > 0$,~$0\leq |c_{i}|^{2}\leq 1$ and $\sum_{i=1}^{d}|c_{i}|^{2}=1$, we will  prove 
 \begin{equation*}
    \sum_{i=1}^{d} \frac{|c_i|^2}{|c_i|^2 + \lambda}\leq \frac{1}{\lambda+1/d} 
 \end{equation*}
   and the equality holds if and only if $|c_{i}|^{2}=1/d$ for all $i$.  

   \begin{proof}
   Let $x_i = |c_i|^2$, then $x_i\in [0,1]$ and $\sum_{i=1}^d x_i = 1$. For $x\in[0,1]$, the function $g(x) = \frac{x}{x + \lambda}$ is strictly concave because its second derivative is
$g''(x) = -\frac{2\lambda}{(x + \lambda)^3} < 0$.
  Since  $g(x)$ is a concave function, for any weights  $0\leq p_i\leq 1$  satisfying  $\sum_{i=1}^{d} p_i = 1$, apply Jensen’s inequality, we have  
\begin{align*}
\sum_{i=1}^{d}p_{i}g(x_{i})\leq g(\sum_{i=1}^{d}p_{i}x_{i}).
\end{align*}
Set $p_{i}=\frac{1}{d}$, under the constraint  $\sum_{i=1}^{d} x_i = 1$, then it has
\begin{equation*}
    \frac{1}{d}\sum_{i=1}^{d}g(x_{i})\leq 
    g(\frac{1}{d}).
\end{equation*}
Therefore, 
\begin{equation*}
    \sum_{i=1}^{d}g(x_{i})\leq \frac{1}{\lambda+1/d} ,
\end{equation*}
where the equality holds if and only if $x_1=x_2=\cdots=x_d=1/d$.
\end{proof}

\appendixsection{The proof of lemma \ref{Eq. vec}}\label{app:D}

In order to maximize
$S = \sum_{i=1}^{d} \left| x_i - \frac{1}{d} \right|$, we divide the indices into two categories
$I = \{ i \mid x_i \geq \frac{1}{d} \}$ and 
$J = \{ i \mid x_i < \frac{1}{d} \}$.
Then we have $ S = \sum_{i \in I} \left( x_i - \frac{1}{d} \right) + \sum_{i \in J} \left( \frac{1}{d} - x_i \right).$
Let $|I| = k$, then $|J| = d - k$, this yields to
\begin{eqnarray*}
    S &=& \left( \sum_{i \in I} x_i - \frac{k}{d} \right) + \left( \frac{d - k}{d} - \sum_{i \in J} x_i \right)\\
    &=& 2 \sum_{i \in I} x_i - 1 + \frac{d - 2k}{d},\\
    &\leq &  \frac{2(d-1)}{d},
\end{eqnarray*}
employing
$\sum_{i \in I} x_i + \sum_{i \in J} x_i = 1$.
The maximum is attained when
$k = 1$, where $x_i = 1$ for some $i$.
 
\appendixsection{The proof of Theorem \ref{pro-three}}\label{app:E}

For any quantum state $\rho$, $\mathcal{W_{\rm{tr}}}(\rho)$ and $\mathcal{P_{\rm{tr}}}(\rho)$ are both convex, so we only need to prove Theorem \ref{pro-three} for pure state. Next, we consider  quantum state $\rho$ as any pure state.

First, we suppose $\rho_{\text{diag}} = \mathrm{diag}(p_1, p_2, p_3)$ with diagonal entriess $p_i$ of $\rho$ in nonincreasing order,  $p_i \geq 0$ for all $i$, $\sum_i p_i = 1$. 
Then $\mathcal{P_{\rm{tr}}}(\rho)^{2}$ is given by
\[
\mathcal{P_{\rm{tr}}}(\rho)^{2} = \frac{9}{16} \left\| \rho_{\text{diag}} - \frac{I}{3} \right\|_{\mathrm{tr}}^{2} = 
\begin{cases}
\dfrac{9}{4}\left(p_1 - \dfrac{1}{3}\right)^2, & \text{if } p_2 {\leq}\dfrac{1}{3}, \\[8pt]
\dfrac{9}{4}\left(p_3 - \dfrac{1}{3}\right)^2, & \text{if } p_2 > \dfrac{1}{3}.
\end{cases}
\]

Second, we suppose $\lambda_{1}, \lambda_{2}, \lambda_{3}$ are the eigenvalues of $\rho - \rho_{\text{diag}}$  in non-increasing order. It follows from Lemma \ref{Eq. lem-im} that  
\[
\mathcal{W_{\rm{tr}}}(\rho)^{2} = \frac{9}{16} ||\rho - \rho_{\text{diag}}||_{\rm{tr}}^{2} = \frac{9}{4} \lambda_{1}^{2}.
\]
Note that
\begin{eqnarray*}
    {\rm{Tr}}(\rho-\rho_{\text{diag}})^{2}
    &=&{\rm{Tr}}(\rho^{2})-2{\rm{Tr}}(\rho\rho_{\text{diag}})+{\rm{Tr}}(\rho_{\text{diag}}^{2})\\\nonumber
    &=&1-\sum_{i=1}^{3}p_{i}^{2}\\\nonumber
    &=&2S
\end{eqnarray*}
with 
$S=p_{1}p_{2}+p_{2}p_{3}+p_{3}p_{1}.$
Combined with the relation 
\begin{equation*}
    {\rm{Tr}}(\rho-\rho_{\text{diag}})^{2}= \sum_{i}\lambda_{i}^{2},
\end{equation*}
we derive 
\begin{equation*}
    2S=\sum_{i=1}^{3}\lambda_{i}^{2}\geq  \frac{3}{2}\lambda_{1}^{2}
\end{equation*}
by Cauchy-Schwarz inequality $ \lambda_{2}^{2}+ \lambda_{3}^{2}\geq \frac{( \lambda_{2}+ \lambda_{3})^{2}}{2}=\frac{ \lambda_{1}^{2}}{2}$, which implies an upper bound for $\mathcal{W_{\rm{tr}}}(\rho)^{2}$ as
\begin{equation*}
    \mathcal{W_{\rm{tr}}}(\rho)^{2} \leq {3S}.
\end{equation*}

Next, we analyze the maximum of  $\mathcal{W_{\rm{tr}}}(\rho)^{2} + \mathcal{P_{\rm{tr}}}(\rho)^{2}$. 
If $p_{2} \leq \frac{1}{3}$,
by the analysis above, we get
\begin{equation*}
    \mathcal{W_{\rm{tr}}}(\rho)^{2} + \mathcal{P_{\rm{tr}}}(\rho)^{2} 
\leq \frac{9}{4} \left( \frac{4}{3}S + \left(p_1 - \frac{1}{3}\right)^{2} \right).
\end{equation*}
{Using the Lagrange multiplier method, we find that the extremum of the function $f(p_1, p_2, p_3)=\frac{4}{3}S+ (p_{1}-\frac{1}{3})^{2}$ under the constraints $\sum_{i=1}^{3}p_{i}=1$  and  $p_{i}\geq 0$ for $i=1,2,3$ is $\frac{4}{9}$, which is attained at the extremum points $(p_1,p_2,p_3)=\{(1-2t,t,t)\mid t\in[0,\frac{1}{3}]\}$}.
Therefore we get the wave-particle duality 
\begin{equation*}
    \mathcal{W_{\rm{tr}}}(\rho)^{2} + \mathcal{P_{\rm{tr}}}(\rho)^{2} \leq1,
\end{equation*}
which is attained if and only if 
$\lambda_{2} = \lambda_{3}$, and  $(p_{1}, p_{2}, p_{3}) = (1-2t, t, t)$ with $t \in [0, \frac{1}{3}]$. 
The analysis is also true for the case  $p_{2} > \frac{1}{3}$.

Now we search for the quantum state satisfying 
the equality $\mathcal{W_{\rm{tr}}}(\rho)^{2} + \mathcal{P_{\rm{tr}}}(\rho)^{2} = 1$. 
For any $3 \times 3$ matrix $\rho-\rho_{\text{diag}}$, its characteristic equation is given by  
\begin{equation}\label{eq app 3-1}
x^3 - c_1 x^2 + c_2 x - c_3 = 0,
\end{equation}
where $c_1, c_2, c_3$ are the elementary symmetric polynomials in the eigenvalues $\lambda_1, \lambda_2, \lambda_3$:  
\begin{eqnarray*}
c_1 &&= \lambda_1 + \lambda_2 + \lambda_3, \\
c_2 &&= \lambda_1\lambda_2 + \lambda_2\lambda_3 + \lambda_3\lambda_1, \\
c_3 &&= \lambda_1\lambda_2\lambda_3.
\end{eqnarray*}
Employing the relations 
\begin{eqnarray*}
\text{Tr}(\rho-\rho_{\text{diag}}) &&= \sum_{i=1}^3 \lambda_i=0,\\
 \text{Tr}(\rho-\rho_{\text{diag}})^2 &&= 1 - \sum_i p_i^2=  \sum_{i=1}^3 \lambda_i^2 ,\\
\text{Tr}(\rho-\rho_{\text{diag}})^3 &&=1 - 3 \sum_i p_i^2 + 2 \sum_i p_i^3= \sum_{i=1}^3 \lambda_i^3,
\end{eqnarray*}
we obtain the coefficients $c_i$ can be expressed by $p_i$ ($i=1,2,3$) as 
\begin{eqnarray*}
c_1 &&=0,\\
c_2 &&= -\frac{1}{2}( 1 - \sum_i p_i^2),\\
 c_3 &&=\frac{1}{3} (1 - 3 \sum_i p_i^2 + 2 \sum_i p_i^3).
\end{eqnarray*}
 Substituting $c_i$ into 
the characteristic equation Eq. (\ref{eq app 3-1}) as well as $p_1 = 1 - 2t, \, p_2 = t, \, p_3 = t $,
we get a cubic equation of $t$,
\begin{equation}\label{eq app 3}
x^{3} - (2t - 3t^{2})x - 2t^{2}(1 - 2t) = 0.
\end{equation}
The condition $\lambda_2=\lambda_3$ is equivalent to  the cubic equation has a double root. This demands the coefficients in Eq. (\ref{eq app 3}) satisfy $t = 0$ or $t = \frac{1}{3}$. Therefore, we derive that the quantum states satisfying the 
equality $\mathcal{W_{\rm{tr}}}(\rho)^{2} + \mathcal{P_{\rm{tr}}}(\rho)^{2} = 1$ are {$\rho = |\psi\rangle\langle\psi|$ with $|\psi\rangle = \frac{1}{\sqrt{3}}(e^{i\theta_{1}},e^{i\theta_{2}},e^{i\theta_{3}})^{\top}$, or $\rho$ is basis-state, $\rho = |i\rangle\langle i|$.

\end{document}